\documentclass[11pt,letterpaper]{article}
\usepackage[margin=1in]{geometry}
\usepackage{ifpdf}
\ifpdf
    \usepackage[pdftex]{graphicx}
    \graphicspath{{figs/}{matlab/}}   
    \usepackage[update]{epstopdf}
\else
	\usepackage{graphicx}
	\graphicspath{{figs/}{matlab/}}   
\fi
\usepackage{pdflscape}
\usepackage{wrapfig,bbm}
\usepackage{enumitem,color}
\usepackage{amsthm}
\newtheorem{theorem}{Theorem}

\usepackage{hyperref}
\usepackage{array}
\usepackage{amsmath,amsthm}
\usepackage{amssymb}
\usepackage{amstext}
\usepackage{cite,setspace}
\usepackage{multirow}
\newtheorem{prop}[theorem]{\textbf{Proposition}}

\begin{document}

\title{A Note on the Rate Region of Exact-Repair Regenerating Codes}
\author{Chao Tian}
\maketitle

\begin{abstract}
The rate region of the $(5,4,4)$ exact-repair regenerating codes is provided. The outer bound is obtained through extension of the computational approach developed in an earlier work, and this region is indeed achievable using the canonical layered codes. This result is part of the online collection of ``Solutions of Computed Information Theoretic Limits (SCITL)''. 
\end{abstract}

\section{Introduction}

The precise storage-repair-bandwidth tradeoff, or the rate region, of exact-repair regenerating codes turns out to be more difficult to characterize than that of the functional-repair version, the latter of which has been known for a few years \cite{Dimakis:10}. In a recent work \cite{Tian:JSAC13}, the author provided a characterization of the $(n,k,d)=(4,3,3)$ exact-repair regenerating codes, {\em i.e.,} when there are $n=4$ storage nodes, any $k=3$ of the nodes can completely recover the stored data, and any failed node can be repaired by the remaining $d=3$ nodes. This result conclusively answered in the affirmative the question whether there is a material separation between the tradeoff of the exact-repair version and  that of the functional-repair version. The converse proof in this characterization was obtained using a less conventional computational approach, based on a strategic application of Yeung's linear programming formulation for information inequalities \cite{Yeung:97}. 

Several analytically derived outer bounds for general exact-repair regenerating codes, or for the restricted setting of {\em linear} exact-repair regenerating codes, were later discovered in recent works \cite{Sasidharan:14, Duursma:14,Prakash:15, Mohajer:15}. In the restricted linear code setting, the rate region for the $(n,k=n-1,d=n-1)$ codes was given in \cite{Prakash:15}. For general coding functions, a partial characterization was given in \cite{Mohajer:15} for the $(5,4,4)$ codes. The complete characterization of the fundamental tradeoff for the $(5,4,4)$ case is however not yet available. In this short note, we provide the characterization of the $(5,4,4)$ exact-repair regenerating codes, the converse of which is obtained using the computational approach developed in \cite{Tian:JSAC13}. The proof is given as tabulation without a  translation into the conventional chains of inequalities. 

Though this particular piece of result was obtained by the author much earlier, the motivation to put it in writing was a recent conversation with Dr. Tie Liu, who pointed out that its availability in public domain may be helpful for future researchers, though from the perspective of developing the computational approach it is an extension of \cite{Tian:JSAC13} with more variables and constraints. This characterization and outer bounds for other cases of exact-repair regenerating codes ({\em e.g.,} the $(5,3,4)$ case), as well as bounds for several other storage and communication problems, are (or will be) included in the online collection of ``Solutions of Computed Information Theoretic Limits (SCITL)'' hosted at \cite{TianWebpage}.

\section{The Rate Region of $(5,4,4)$ Regenerating Codes}
\label{sec:main}

The main result of this note is the following theorem, where $\bar{\alpha}$  and $\bar{\beta}$ are the per-node storage capacity and per-helper repair bandwidth, respectively, which are normalized by the total amount of data $B$.

\begin{theorem}
The rate region of the $(5,4,4)$ exact-repair codes is the collection of $(\bar\alpha,\bar\beta)$ pairs satisfying the following conditions,
\begin{align}
4\bar{\alpha}\geq 1,\quad 3\bar\alpha+\bar\beta\geq 1, \quad 15\bar\alpha+10\bar\beta\geq 6,\quad 5\bar\alpha+10\bar\beta\geq 3,\quad 10\bar{\beta}\geq 1. \label{eqn:bounds}
\end{align}
\end{theorem}

\begin{figure}
\centering
\includegraphics[width=10cm]{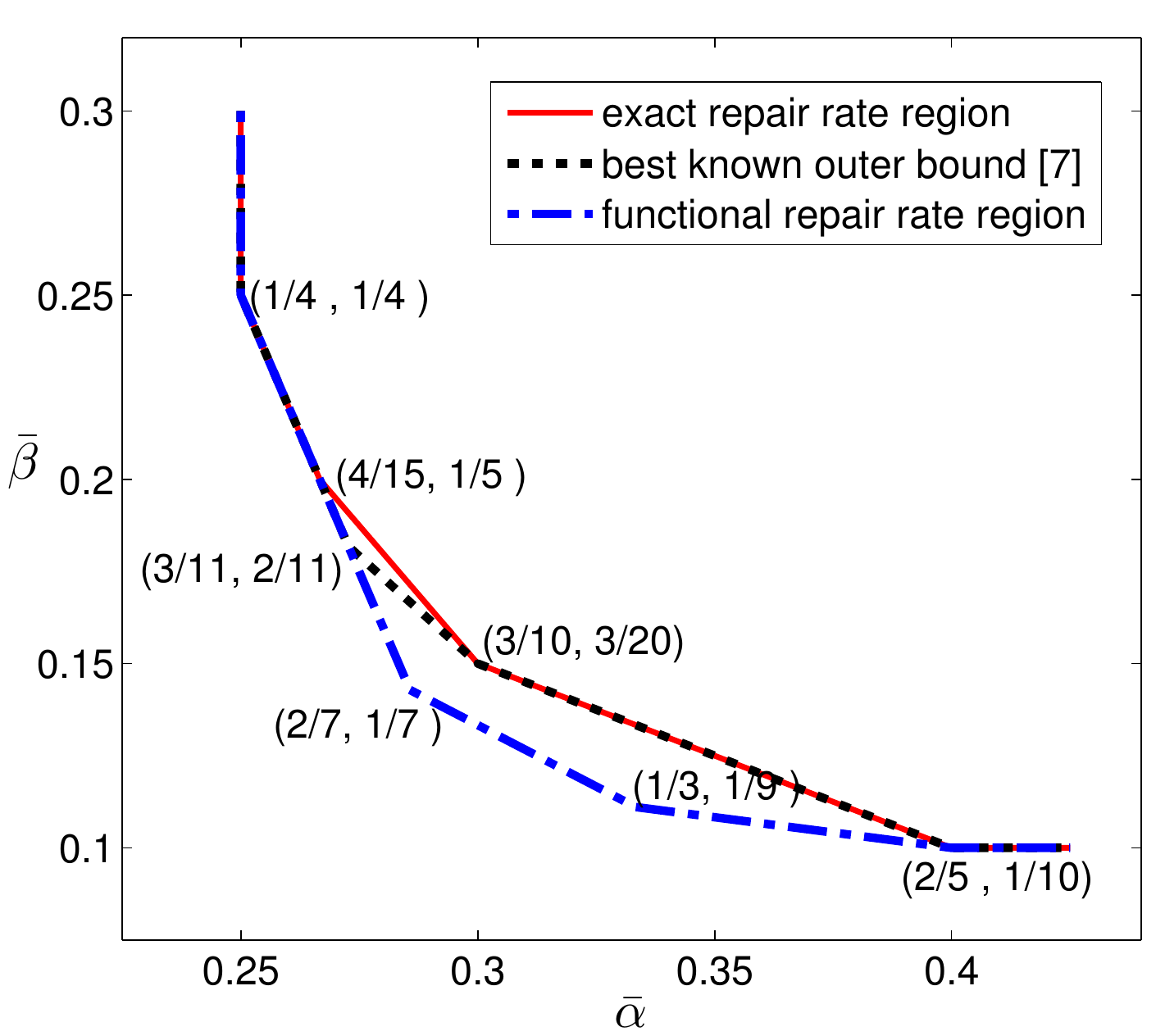}
\caption{Rate region of $(5,4,4)$ exact-repair regenerating codes. \label{fig:region}}
\end{figure}

This region is illustrated in Fig. \ref{fig:region}. It can be verified straightforwardly that this region is achievable using the canonical layered codes proposed in \cite{Tian:14layered}. The first, the second and the last inequalities in (\ref{eqn:bounds}) are previously known by viewing the functional-repair tradeoff as an outer bound to the exact-repair case. The remaining two are proved in the next section. 

This region is the same as identified in \cite{Prakash:15} for the restricted setting of linear regenerating codes, and thus there is no loss of optimality by considering only linear codes in this case. 

\section{A Converse Proof by Tabulation}

We use ${\alpha}$  and ${\beta}$ to denote the per-node storage capacity and per-helper repair bandwidth, respectively, before normalization by the data size $B$. The (random) helper information sent from node $i$ to node $j$ is denoted as $S_{i\rightarrow j}$ and the (random) information stored on node $i$ is denoted as $W_i$. Due to the symmetry of the problem \cite{Tian:JSAC13}, we can restrict the proof to symmetric codes without loss of optimality. 

The two inequalities are presented as two propositions, in the form of tabulation. Each inequality is given as two tables, the first of which lists the joint entropy terms in the proof, and the second of which gives the known sub-modular inequalities (Shannon-type inequalities) and the coefficients of these inequalities, whose last row, as the summation of all the other rows, is exactly the sought-after inequality. 
Note that each row in the second table is a sub-modular inequality (or an integer multiple of such an inequality), possibly after permutation of the indices for each entropy term. 

These tables are given in the form obtained almost directly from the computation without much human simplification. Thus in a sense they are the ``raw data'', and certain steps can be taken to make it more concise and human-friendly, the automation of which is part of our ongoing work. This is also the motivation to establish SCITL online collection, where the matrix version of the solutions can be accessed such that further data analysis can be done more conveniently, and future researchers can interpret the data in manners more meaningful to them. In addition, for more complex problems the tabulation we use here may become too cumbersome to present, and a data file is a more appropriate media.

\begin{prop}
\label{prop:bound1}
Any $(5,4,4)$ exact-repair codes must satisfy $15\alpha+10\beta\geq 6B$.
\end{prop}
\begin{proof}
See Table \ref{tab:correspondence} and Table \ref{table:cancellation}.
\end{proof}

\begin{prop}
\label{prop:bound2}
Any $(5,4,4)$ exact-repair codes must satisfy $5\alpha+10\beta\geq 3B$.
\end{prop}
\begin{proof}
See Table \ref{tab:correspondence2} and Table \ref{table:cancellation2}.
\end{proof}

\begin{table}[ct]
\begin{center}
\caption{The entropy terms used in the proof of Proposition \ref{prop:bound1}.}
\label{tab:correspondence}
\begin{tabular}{|c|c|}
\hline
$T_{ 1}$ & $\beta$ \\
$T_{ 2}$ & $\alpha$ \\
$T_{ 3}$ & $H(S_{5\rightarrow4},W_{1})$ \\
$T_{ 4}$ & $H(S_{5\rightarrow2},S_{4\rightarrow5},W_{2})$ \\
$T_{ 5}$ & $H(S_{5\rightarrow3},S_{4\rightarrow3},W_{1})$ \\
$T_{ 6}$ & $H(S_{5\rightarrow2},S_{4\rightarrow3},S_{3\rightarrow5},W_{2})$ \\
$T_{ 7}$ & $H(S_{5\rightarrow2},S_{4\rightarrow5},S_{4\rightarrow2},S_{3\rightarrow5},W_{1})$ \\
$T_{ 8}$ & $H(W_{2},W_{1})$ \\
$T_{ 9}$ & $H(S_{5\rightarrow4},W_{2},W_{1})$ \\
$T_{10}$ & $H(S_{5\rightarrow2},W_{2},W_{1})$ \\
$T_{11}$ & $H(S_{5\rightarrow2},S_{4\rightarrow5},W_{2},W_{1})$ \\
$T_{12}$ & $H(S_{5\rightarrow2},S_{4\rightarrow3},S_{3\rightarrow5},W_{2},W_{1})$ \\
$T_{13}$ & $H(S_{4\rightarrow5},W_{5},W_{2},W_{1})$ \\
$T_{14}$ & $H(S_{5\rightarrow3},S_{4\rightarrow3},W_{3},W_{2},W_{1})$ \\
$T_{15}$ & $B$ \\
\hline
\end{tabular}
\end{center}
\end{table}

\begin{table*}[tcb]
\setlength{\tabcolsep}{4pt}
\begin{center}
\caption{Proof of Proposition \ref{prop:bound1} with terms defined in Table \ref{tab:correspondence}.}
\label{table:cancellation}
\begin{tabular}{|ccccc ccccc ccccc|}
\hline
$\beta$  &$\alpha$  &$T_{ 3}$  &$T_{ 4}$  &$T_{ 5}$  &$T_{ 6}$  &$T_{ 7}$  &$T_{ 8}$  &$T_{ 9}$  &$T_{10}$  &$T_{11}$  &$T_{12}$  &$T_{13}$  &$T_{14}$  &$B$  \\
\hline
$  5$     &$  5$     &$ -5$     &          &          &          &          &          &          &          &          &          &          &          &           \\
          &$  1$     &$  1$     &          &          &          &          &          &$ -1$     &          &          &          &          &          &           \\
          &$  4$     &          &          &$  4$     &          &          &          &          &          &          &          &          &$ -4$     &           \\
$  5$     &          &$  5$     &          &$ -5$     &          &          &          &          &          &          &          &          &          &           \\
          &$  2$     &          &          &$  2$     &          &          &          &          &          &          &          &$ -2$     &          &           \\
          &          &          &          &          &          &          &          &          &          &$ -2$     &          &$  2$     &$  2$     &$ -2$     \\
          &$  2$     &          &          &          &          &          &$ -1$     &          &          &          &          &          &          &           \\
          &          &          &$ -1$     &          &$  1$     &          &          &          &          &$  1$     &$ -1$     &          &          &           \\
          &          &          &          &          &$ -1$     &          &          &          &          &$  1$     &          &          &$  1$     &$ -1$     \\
          &          &          &          &$ -1$     &          &$  1$     &          &$  1$     &          &          &          &          &          &$ -1$     \\
          &$  1$     &          &          &          &          &          &          &          &$  1$     &          &          &          &          &$ -1$     \\
          &          &$ -1$     &$  1$     &          &          &          &$  1$     &          &$ -1$     &          &          &          &          &           \\
          &          &          &          &          &          &$ -1$     &          &          &          &          &$  1$     &          &$  1$     &$ -1$    \\
\hline     
\hline
$10$ &$15$  &        &        &        &         &       &        &        &        &       &             &              &     &  -6   \\\hline
\end{tabular}
\end{center}
\end{table*}

\begin{table}[ct]
\begin{center}
\caption{The entropy terms used in the proof of Proposition \ref{prop:bound2}.}
\label{tab:correspondence2}
\begin{tabular}{|c|c|}
\hline
$T_{ 1}$ & $\beta$ \\
$T_{ 2}$ & $H(S_{5\rightarrow3},S_{4\rightarrow3})$ \\
$T_{ 3}$ & $H(S_{5\rightarrow4},S_{5\rightarrow3},S_{4\rightarrow3})$ \\
$T_{ 4}$ & $H(S_{5\rightarrow4},S_{4\rightarrow5},S_{3\rightarrow5})$ \\
$T_{ 5}$ & $H(S_{5\rightarrow4},S_{5\rightarrow2},S_{4\rightarrow2},S_{3\rightarrow2})$ \\
$T_{ 6}$ & $H(S_{5\rightarrow4},S_{5\rightarrow1},S_{4\rightarrow1},S_{3\rightarrow1},S_{2\rightarrow1})$ \\
$T_{ 7}$ & $\alpha$ \\
$T_{ 8}$ & $H(S_{5\rightarrow4},W_{1})$ \\
$T_{ 9}$ & $H(S_{5\rightarrow2},W_{2})$ \\
$T_{10}$ & $H(S_{5\rightarrow2},S_{4\rightarrow2},W_{2})$ \\
$T_{11}$ & $H(S_{5\rightarrow4},S_{4\rightarrow5},W_{1})$ \\
$T_{12}$ & $H(S_{5\rightarrow3},S_{4\rightarrow3},W_{1})$ \\
$T_{13}$ & $H(S_{5\rightarrow2},S_{4\rightarrow5},W_{2})$ \\
$T_{14}$ & $H(S_{5\rightarrow2},S_{4\rightarrow2},S_{3\rightarrow5},W_{2})$ \\
$T_{15}$ & $H(S_{5\rightarrow4},S_{5\rightarrow2},S_{4\rightarrow2},W_{2})$ \\
$T_{16}$ & $H(S_{5\rightarrow4},S_{5\rightarrow3},S_{4\rightarrow3},W_{1})$ \\
$T_{17}$ & $H(S_{5\rightarrow3},S_{4\rightarrow5},S_{4\rightarrow3},S_{3\rightarrow5},W_{1})$ \\
$T_{18}$ & $H(W_{2},W_{1})$ \\
$T_{19}$ & $H(S_{5\rightarrow2},W_{2},W_{1})$ \\
$T_{20}$ & $H(S_{5\rightarrow2},S_{4\rightarrow2},W_{2},W_{1})$ \\
$T_{21}$ & $H(S_{5\rightarrow4},S_{4\rightarrow5},W_{2},W_{1})$ \\
$T_{22}$ & $H(S_{5\rightarrow2},S_{4\rightarrow5},W_{2},W_{1})$ \\
$T_{23}$ & $H(S_{5\rightarrow1},S_{4\rightarrow2},W_{2},W_{1})$ \\
$T_{24}$ & $H(S_{5\rightarrow2},S_{4\rightarrow2},S_{3\rightarrow2},W_{1})$ \\
$T_{25}$ & $H(S_{5\rightarrow2},S_{4\rightarrow2},S_{3\rightarrow5},W_{2},W_{1})$ \\
$T_{26}$ & $H(S_{5\rightarrow2},S_{5\rightarrow1},S_{4\rightarrow2},S_{3\rightarrow2},W_{1})$ \\
$T_{27}$ & $H(S_{5\rightarrow4},S_{5\rightarrow2},S_{4\rightarrow2},S_{3\rightarrow2},W_{1})$ \\
$T_{28}$ & $H(S_{5\rightarrow2},S_{5\rightarrow1},S_{4\rightarrow5},S_{4\rightarrow2},S_{3\rightarrow2},W_{1})$ \\
$T_{29}$ & $H(S_{4\rightarrow5},W_{5},W_{2},W_{1})$ \\
$T_{30}$ & $H(S_{5\rightarrow3},S_{4\rightarrow3},W_{2},W_{1})$ \\
$T_{31}$ & $H(S_{5\rightarrow3},S_{5\rightarrow2},S_{4\rightarrow3},W_{2},W_{1})$ \\
$T_{32}$ & $H(S_{5\rightarrow3},S_{5\rightarrow2},S_{4\rightarrow3},S_{4\rightarrow2},W_{2},W_{1})$ \\
$T_{33}$ & $B$ \\
\hline
\end{tabular}
\end{center}
\end{table}

\begin{landscape}
\begin{table*}[tcb]
\setlength{\tabcolsep}{2.7pt}
\begin{center}
\caption{Proof of Proposition \ref{prop:bound2} with terms defined in Table \ref{tab:correspondence2}.}
\label{table:cancellation2}
\begin{tabular}{|ccccc ccccc ccccc ccccc ccccc ccccc ccc|}
\hline
$\beta$  &$T_{ 2}$  &$T_{ 3}$  &$T_{ 4}$  &$T_{ 5}$  &$T_{ 6}$  &$\alpha$  &$T_{ 8}$  &$T_{ 9}$  &$T_{10}$  &$T_{11}$  &$T_{12}$  &$T_{13}$  &$T_{14}$  &$T_{15}$  &$T_{16}$  &$T_{17}$  &$T_{18}$  &$T_{19}$  &$T_{20}$  &$T_{21}$  &$T_{22}$  &$T_{23}$  &$T_{24}$  &$T_{25}$  &$T_{26}$  &$T_{27}$  &$T_{28}$  &$T_{29}$  &$T_{30}$  &$T_{31}$  &$T_{32}$  &$B$  \\
\hline
$  1$     &          &          &          &          &          &          &          &          &          &          &$  1$     &          &          &          &          &          &          &          &$ -1$     &          &          &          &          &          &          &          &          &          &          &          &          &           \\
          &$  9$     &          &          &          &          &$  9$     &          &          &          &          &$ -9$     &          &          &          &          &          &          &          &          &          &          &          &          &          &          &          &          &          &          &          &          &           \\
$ 22$     &$-11$     &          &          &          &          &          &          &          &          &          &          &          &          &          &          &          &          &          &          &          &          &          &          &          &          &          &          &          &          &          &          &           \\
          &          &          &          &          &          &          &$ -1$     &          &          &          &$  2$     &          &          &          &          &          &          &          &          &          &          &          &$ -1$     &          &          &          &          &          &          &          &          &           \\
          &          &          &          &          &          &          &$ -3$     &          &          &          &$  6$     &          &          &          &          &          &          &          &$ -3$     &          &          &          &          &          &          &          &          &          &          &          &          &           \\
          &$  3$     &          &          &          &          &          &          &          &          &          &          &          &          &          &          &          &$  3$     &          &          &          &          &          &          &          &          &          &          &          &$ -3$     &          &          &           \\
          &          &          &          &          &          &          &          &          &          &          &          &          &$ -2$     &          &          &          &          &          &$  2$     &          &          &$  2$     &          &          &          &          &          &$ -2$     &          &          &          &           \\
$ -1$     &          &$  1$     &          &          &          &          &$  1$     &          &          &          &          &          &          &          &$ -1$     &          &          &          &          &          &          &          &          &          &          &          &          &          &          &          &          &           \\
          &          &          &          &          &          &          &          &          &$ -2$     &          &          &          &$  2$     &          &          &          &          &          &$  2$     &          &          &          &          &$ -2$     &          &          &          &          &          &          &          &           \\
          &          &          &          &          &          &          &          &          &          &          &          &          &          &          &          &          &          &$ -2$     &          &          &$  2$     &          &          &          &          &          &          &$  2$     &          &          &          &$ -2$     \\
          &          &          &          &          &          &          &          &          &          &          &          &          &          &          &          &          &          &          &          &          &$ -2$     &          &          &$  2$     &          &          &          &          &$  2$     &$ -2$     &          &           \\
          &          &          &          &          &          &          &          &          &          &          &          &          &          &          &          &          &$ -2$     &$  4$     &          &          &          &$ -2$     &          &          &          &          &          &          &          &          &          &           \\
$ -1$     &          &          &          &          &          &$  1$     &          &$  1$     &          &          &          &          &          &          &          &          &$ -1$     &          &          &          &          &          &          &          &          &          &          &          &          &          &          &           \\
          &          &$ -1$     &          &          &          &          &$  1$     &          &          &          &          &          &          &$  1$     &          &          &          &$ -1$     &          &          &          &          &          &          &          &          &          &          &          &          &          &           \\
          &          &          &$ -1$     &          &          &          &          &          &$  1$     &$  1$     &          &          &          &          &          &          &          &$ -1$     &          &          &          &          &          &          &          &          &          &          &          &          &          &           \\
          &          &          &          &          &$ -1$     &          &          &          &          &          &          &          &          &          &          &          &          &          &          &          &          &          &$  1$     &          &$  1$     &          &          &          &$ -1$     &          &          &           \\
          &$ -1$     &          &$  1$     &          &          &          &$  1$     &          &          &          &          &$ -1$     &          &          &          &          &          &          &          &          &          &          &          &          &          &          &          &          &          &          &          &           \\
$ -1$     &          &          &          &$  1$     &          &          &$  1$     &          &          &          &          &          &          &          &          &          &          &          &          &          &          &          &          &          &          &$ -1$     &          &          &          &          &          &           \\
          &          &          &          &$ -1$     &$  1$     &          &          &          &          &          &          &          &          &          &$  1$     &          &          &          &          &          &          &          &          &          &          &$ -1$     &          &          &          &          &          &           \\
          &          &          &          &          &          &          &          &          &          &$ -1$     &          &          &          &          &          &$  1$     &          &          &          &$  1$     &          &          &          &          &          &          &          &          &          &          &          &$ -1$     \\
          &          &          &          &          &          &          &          &          &          &          &          &          &          &          &          &          &          &          &          &          &          &          &          &          &$ -1$     &          &$  1$     &          &          &$  1$     &          &$ -1$     \\
          &          &          &          &          &          &          &          &          &          &          &          &          &          &          &          &          &          &          &          &          &          &          &          &          &          &          &$ -1$     &          &          &$  1$     &$  1$     &$ -1$     \\
          &          &          &          &          &          &          &          &          &          &          &          &          &          &          &          &$ -1$     &          &          &          &          &          &          &          &          &          &$  2$     &          &          &          &          &$ -1$     &           \\
          &          &          &          &          &          &          &          &          &          &          &          &          &          &          &          &          &          &          &          &$ -1$     &          &          &          &          &          &          &          &          &$  2$     &          &          &$ -1$     \\
          &          &          &          &          &          &          &          &$ -1$     &$  1$     &          &          &$  1$     &          &$ -1$     &          &          &          &          &          &          &          &          &          &          &          &          &          &          &          &          &          &           \\
\hline
\hline
$20$      &          &          &          &          &          &     $10$     &          &   &     &          &          &     &          &     &          &          &          &          &          &          &          &          &          &          &          &          &          &          &          &          &          &       -6    \\
\hline
\end{tabular}
\end{center}
\end{table*}
\end{landscape}

\section{Conclusion}

The rate region of the $(5,4,4)$ exact-repair regenerating codes is characterized in this note. This is part of the online collection of ``Solutions of Computed Information Theoretic Limits (SCITL)'' hosted at \cite{TianWebpage}, which hopefully in the future can serve as a data depot for information theoretic limits obtained through computational approaches. Several results in this collection requires non-trivial generalization or variation of the approach outlined in \cite{Tian:JSAC13}, the details of which will be presented elsewhere. We welcome contributions from researchers in the field who have developed information-theoretic bounds using computational approaches, either through approaches similarly to or completely different from that in \cite{Tian:JSAC13}. 

\bibliographystyle{IEEEbib}

\end{document}